\documentclass[5p,times,twocolumn]{elsarticle}
\usepackage[T1]{fontenc}
\usepackage{amsmath,amssymb,amsthm}
\usepackage{graphicx}
\usepackage{microtype}
\usepackage{xcolor}
\usepackage{booktabs}
\usepackage{array}
\usepackage[colorlinks=true,linkcolor=blue!50!black,citecolor=blue!50!black,urlcolor=blue!50!black]{hyperref}
\graphicspath{{figures/}}
\journal{Physica A: Statistical Mechanics and its Applications}
\biboptions{sort&compress}

\newtheorem{theorem}{Theorem}
\newtheorem{proposition}[theorem]{Proposition}
\newtheorem{corollary}[theorem]{Corollary}
\newtheorem{definition}[theorem]{Definition}
\newtheorem{remark}[theorem]{Remark}
\newcommand{\dd}{\mathrm{d}}
\newcommand{\PP}{\mathcal{P}}
\newcommand{\SF}{\mathcal{S}}
\newcommand{\Wone}{W_1}
\newcommand{\DF}{D_F}

\begin{document}
\begin{frontmatter}

\title{Response-Aware Coarse-Graining in Summary-Mediated Opinion Dynamics}

\author{Ruben E. Ara\'ujo\corref{cor1}}
\ead{rubenesteche@hotmail.com}
\cortext[cor1]{Corresponding author.}
\address{Independent Researcher, Brazil}

\begin{abstract}
Interacting populations increasingly respond to compressed representations
generated from their own collective state. We ask which information such an
endogenous coarse-graining must preserve in order to reproduce the collective
dynamics. The question has two complementary answers for continuous opinions
with attraction, indifference, and repulsion. Inside any strict $K$-bloc
response cell where every represented stance attracts every bloc, the dynamics
reduce exactly to translation and disagreement. The represented distribution's
shape cancels from all disagreement modes, which form a homogeneously anchored
weighted signed-Laplacian system, while represented-mean error acts only on
translation. This reduction yields a geometry-dependent stability boundary
and shows that prevalence bias produces secular drift at marginal stability
only through its projection onto the neutral eigenspace. Outside that common
linear-response regime, threshold geometry makes distributional information
observable. The induced force discrepancy is a behavior-defined integral
probability metric; restricted force probes identify generic finite empirical
discussions, and for finitely supported candidates uniform force error controls
total-variation and Wasserstein error while support reduction incurs a nonzero
distortion floor. A cumulative-
discrepancy characterization further explains why transport-small summaries
can remain behaviorally inaccurate at response thresholds. Mode-resolved
integrations, projected particle simulations, finite-size ensembles, and
threshold-crossing experiments test these two regimes. As a deployment
consequence, broadcast summary error survives as collective noise whereas
independently personalized error self-averages. The model is stylized and not
empirically calibrated; generated discussion summaries motivate the general
problem of endogenous, response-aware coarse-graining.
\end{abstract}


\end{frontmatter}

\section{Introduction}
\label{sec:introduction}

Statistical physics provides a macroscopic language for collective social
states emerging from local interaction rules \citep{castellano2009,lorenz2007,
starnini2026}. In continuous-opinion models, attraction, bounded confidence,
repulsion, external fields, and noise produce consensus, fragmentation,
polarization, and switching between them \citep{deffuant2000,hegselmann2002,
jager2005,sabinmiller2020,goddard2022,bernardo2024}. Media are usually represented as a
fixed field, a prescribed distribution, or an external agent
\citep{mckeown2006,vazmartins2010,pineda2015,pansanella2023}. A different
feedback arises when the signal shown to the population is itself a compressed
representation of the population's current state. The field is then neither
external nor fixed: microscopic opinions are coarse-grained, returned to the
same population, and coarse-grained again after the population responds.
This system-performed coarse-graining and its feedback onto component behavior
belong to the broader framework of endogenous coarse-graining and downward
causation developed by \citet{flack2017}.

Generated summaries make this feedback concrete: an online discussion can be
encountered through a digest that reports a central position, amplifies a
majority, or narrates unequal positions as equally prevalent. Controlled
studies indicate that such summaries can anchor later expression and alter
judgment or participation \citep{govers2026,heseltine2026,shu2026}. The
mathematical problem is broader than AI, however. Related models contain AI
oracles, endogenous aggregators, mediated communication, predictive feedback,
institutional signals, platform-selected information, and hierarchical
coarse-graining \citep{rodrigo2025,li2026,acemoglu2026,tsirtsis2026,wu2026,
segoviamartin2021,candogan2022,widler2026}. Global medians and recommendation
rules have likewise been coupled to collective states
\citep{berenbrink2026,sirbu2019,santos2021}. The distinction studied here is
a distribution-valued endogenous representation processed through a
discontinuous attraction--indifference--repulsion (AIR) response.
Table~\ref{tab:closest-work} states the resulting priority boundary.

\begin{table*}[t]
\centering
\small
\caption{Focused comparison with the closest endogenous-feedback models.
The table separates the feedback representation from the response law and
from the result proved; it is not a general survey of opinion dynamics.}
\label{tab:closest-work}
\begin{tabular}{@{}
  >{\raggedright\arraybackslash}p{0.17\textwidth}
  >{\raggedright\arraybackslash}p{0.20\textwidth}
  >{\raggedright\arraybackslash}p{0.25\textwidth}
  >{\raggedright\arraybackslash}p{0.30\textwidth}@{}}
\toprule
Work & Returned signal & Population response & Principal result closest to the present setting \\
\midrule
\citet{segoviamartin2021} & Reflexive institutional aggregate & Networked preference updating & Convergence effects of proportional-representation feedback \\
\citet{candogan2022} & Platform-selected content & Social learning on stochastic block networks & Large-network reduction to a two-agent system and consensus/disagreement regimes \\
\citet{wu2026} & Co-evolving prediction & Performative predictive loop & Conditions for consensus under prediction--opinion feedback \\
\citet{widler2026} & Hierarchical group and global representations & Community-bounded CODA updating & Multi-level coarse-graining and downward-causation regimes \\
Present work & Endogenous probability measure over represented stances & Discontinuous attraction--indifference--repulsion response & Mode-specific sufficiency inside response cells and bounded-probe identifiability across AIR thresholds \\
\bottomrule
\end{tabular}
\end{table*}

The central question is one of behavioral sufficiency: which properties of a
discussion must a coarse-graining preserve so that every recipient experiences
the same force? Under the linear response
$F(z-x)=z-x$, only the represented mean enters the force, and any
mean-preserving compression is dynamically exact. Attraction--indifference--
repulsion laws instead partition relative opinion into response zones. Moving
even a small amount of mass across a zone boundary can reverse the induced
force. Thus the relevant coarse-graining is not determined by a generic text
similarity or transport score; it depends on the observable through which the
compressed state re-enters the dynamics.

The analysis resolves this question in two regimes. First, inside an arbitrary
open response cell of $K$ weighted blocs, common attraction to the represented
signal makes its mean sufficient for translation and its complete shape
irrelevant to disagreement. The resulting exact reduction is a homogeneously
anchored weighted signed-Laplacian system. The anchored spectral shift is
established network algebra \citep{bronski2014,razaq2025}; the contribution is
deriving that operator from an endogenous distribution-valued channel and
locating representation error exclusively in the translation mode. The
geometry-dependent stability boundary and the neutral-mode criterion for
secular drift follow from this decomposition.

Second, once recipient probes encounter AIR thresholds, moment sufficiency
generally fails and support geometry becomes visible. We show that force
observations on the bounded opinion domain identify generic finite empirical
discussions, establish inverse-stability bounds in total variation and
Wasserstein distance, and derive a nonzero distortion floor for
support-reducing compression. A two-sided decomposition separates transport
error from cumulative mass displaced across response thresholds. These are
bounded-probe inverse results generated by the social response law, related to
finite-rate-of-innovation and sparse-measure recovery
\citep{vetterli2002,candes2014} but with a test class fixed by behavior rather
than chosen for reconstruction.

The two regimes are complementary rather than separate models. The fixed-cell
theorem identifies precisely what is sufficient: represented shape may be
discarded for disagreement, while the represented mean must be retained for
translation. The inverse and distortion results characterize what becomes
observable when that simplification fails. Two-bloc reversal, prevalence-driven translation,
and the contrast between broadcast and personalized summary errors are
consequences used to expose the collective implications. Common-noise and
self-averaging principles themselves are standard
\citep{carmona2018,coculescu2024}; here they diagnose how correlation in the
endogenous representation channel survives or vanishes at population scale.

The model uses projected dynamics on a bounded opinion interval and an
explicit convention at response discontinuities. Every closed-form dynamical
statement is local to a strict response cell and stops at the first threshold
or boundary event unless stated otherwise. The transfer operators are
stylized rather than fitted to a deployed language model or to human response
data. The numerical experiments test the mathematical reductions, connect the
two response regimes, and delimit their robustness; they do not constitute
empirical validation of a particular AI system.

\section{Summary-mediated dynamics}
\label{sec:model}

\subsection{Discussions as measures, summaries as operators}

Let the opinion space be $X=[-1,1]$. A discussion $e$ with stances
$x_j\in X$ is its empirical measure
\begin{equation}
  \mu_e=\frac{1}{|e|}\sum_{j\in e}\delta_{x_j}\in\PP(X).
  \label{eq:empirical}
\end{equation}
Probability-measure descriptions and empirical-measure limits are established
tools in continuous opinion dynamics \citep{como2011,mirtabatabaei2014}. The
measure representation itself is not claimed as new; here it is used so that
the endogenous signal can retain, suppress, or reweight stance distribution.
A deterministic summary operator is a map
\begin{equation}
  \SF_\theta:\PP(X)\longrightarrow\PP(X),
  \qquad Q_e=\SF_\theta[\mu_e],
  \label{eq:summary}
\end{equation}
where $Q_e$ is the distribution of positions perceived in the summary, not a
distribution over strings: a single-stance summary is $Q=\delta_s$, a
``both sides'' narrative a two-point measure, a richer digest a broader
measure. Stances may be latent scores inferred from text, with uncertainty
propagated into $\mu_e$.

Generative outputs are stochastic: a Markov kernel $K_\theta(\mu,\dd Q)$
maps discussions to represented measures. Deterministic operators in the
fidelity results take values in $\PP(X)$. In Section~\ref{sec:stochastic}, a
stance error translates such a measure and may produce an element of
$\PP_1(\mathbb R)$; recipient opinions remain constrained to $X$, and $F$ is
defined on $\mathbb R$. This distinction permits an unbounded Gaussian stress
test without incorrectly treating the translated summary as a measure on
$X$. Because the force is linear in $Q$, its expected one-step value is
generated by the barycenter
$\overline{\SF}_\theta[\mu]=\int Q\,K_\theta(\mu,\dd Q)$. Recursive
trajectories driven by the barycenter and expected trajectories driven by
random $Q$ need not agree.

\subsection{Influence and dynamical distortion}

Let $F(z-x)$ be the response of a recipient at $x$ to represented stance
$z$. For any probability measure $\nu$ on $\mathbb R$ for which the integral
is finite, define
\begin{equation}
  V_F[\nu](x)=\int_{\mathbb R} F(z-x)\,\nu(\dd z)
  \label{eq:force}
\end{equation}
as the induced social force. All deterministic discussion and summary
measures below are supported on $X$; the larger domain is needed only for the
translated stochastic stress test.

\begin{definition}[Dynamical distortion]
For a nonempty recipient set $Y\subseteq X$,
\begin{equation}
  \DF^{Y}(\mu,Q)
  =\sup_{x\in Y}\left|
    \int_X F(z-x)[Q-\mu](\dd z)
  \right|.
  \label{eq:distortion}
\end{equation}
We write $\DF$ when $Y=X$. A summary is dynamically faithful for $(F,Y)$ if
$\DF^{Y}(\mu,\SF[\mu])=0$.
\end{definition}

Equation~\eqref{eq:distortion} is the integral probability metric
\citep{muller1997} generated by the translated family
$\{F(\cdot-x):x\in Y\}$. When that family is measure-determining, $\DF$ is a
metric---in analogy with the injective mean embeddings induced by
characteristic kernels \citep{sriperumbudur2010,muandet2017}---and when it
is not, distinct discussions are dynamically equivalent. The difference from
the machine-learning setting is the origin and role of the test class: $F$
is a behavioral response law fixed by the population rather than a modeling
choice, and the embedding error feeds back into the dynamics that generate
the next input. $\DF$ is an operational error measure for the specified
dynamics, not a replacement for factuality or linguistic quality.

Agents divide a fixed influence budget between direct and represented
exposure, so varying $\lambda$ substitutes between channels without changing
the total interaction scale. Define the unconstrained velocity
\begin{equation}
  B_i(x)=(1-\lambda)V_F[\mu_x](x_i)
  +\lambda V_F[\SF[\mu_x]](x_i),
  \qquad 0\le\lambda\le1.
  \label{eq:raw-velocity}
\end{equation}
The bounded opinion space is enforced by the projected dynamics
\begin{equation}
  \dot x_i=\Pi_{T_X(x_i)}B_i(x),
  \qquad
  \Pi_{T_X(x)}v=
  \begin{cases}
    \max(v,0),&x=-1,\\
    v,&-1<x<1,\\
    \min(v,0),&x=1.
  \end{cases}
  \label{eq:particle}
\end{equation}
Thus outward velocity is suppressed at the boundary while inward velocity is
unchanged. At a discontinuity of the response law we use the Krasovskii
convexification of $B$; equivalently, Eq.~\eqref{eq:particle} is read as the
associated projected differential inclusion. This is a standard convention
for discontinuous bounded-confidence systems \citep{ceragioli2012,altafini2018}.
Every closed-form result below assumes strict response inequalities and stops
before projection activates, so it is independent of the selection on the
switching surface. The full-particle computations use projected Euler steps;
the scalar reduction is checked under temporal refinement. Network structure,
mobility, and heterogeneity are omitted to isolate the summary channel.

Throughout we use the attraction--indifference--repulsion law
\begin{equation}
F(u)=
\begin{cases}
u, & |u|<\epsilon_1,\\
0, & \epsilon_1\le |u|<\epsilon_2,\\
-\eta u, & |u|\ge\epsilon_2,
\end{cases}
\qquad \eta>0,
\label{eq:air}
\end{equation}
with $0<\epsilon_1<\epsilon_2<2$ and jump set
$T=\{\pm\epsilon_1,\pm\epsilon_2\}$, following the social-judgment structure
of assimilation and contrast \citep{jager2005,sabinmiller2020}. The
repulsive branch is a modeling hypothesis, not an empirical universal.

\subsection{Stylized summary operators}
\label{sec:operators}

Mean collapse $\SF_0[\mu]=\delta_{m(\mu)}$, $m(\mu)=\int z\,\mu(\dd z)$,
represents the discussion by one central stance. The contraction family
\begin{equation}
  \SF_c[\mu]=(T_c^\mu)_\#\mu,
  \qquad
  T_c^\mu(z)=m(\mu)+c[z-m(\mu)],
  \label{eq:contraction}
\end{equation}
interpolates to the identity at $c=1$; every $\SF_c$ preserves the mean and
retains a fraction $c^2$ of the variance. For a two-position discussion
$\mu_q=q\delta_{y_+}+(1-q)\delta_{y_-}$ ($y_+>y_-$), prevalence can be
transformed without moving positions:
\begin{align}
  \widehat q_{\rm fb}
  &=\tfrac12+c_q\!\left(q-\tfrac12\right),
  &&0\le c_q<1,
  \label{eq:false-balance}\\
  \widehat q_{\rm maj}
  &=\operatorname{logit}^{-1}
    [a\operatorname{logit}(q)],
  &&a>1,
  \label{eq:majority}
\end{align}
stylizing the narrative and percentage formats of \citet{govers2026}. Both
maps preserve $q=\tfrac12$ and respect the symmetry $q\leftrightarrow1-q$;
the logit map amplifies prevalence odds without introducing a preferred
ideological direction. Consistent with this mechanism, \citet{lei2024} report
that several evaluated
neural opinion summarizers amplify input polarity and under-represent
minority stances, motivating a transfer $q\mapsto\widehat q$ with
$\widehat q>q$ for $q>\tfrac12$. \citet{huang2023} separately show that
semantic similarity does not ensure adequate coverage of opinion diversity.
The composed operator reweights and then contracts about the
\emph{represented} mean $\widehat m=\widehat q\,y_++(1-\widehat q)\,y_-$:
\begin{equation}
  \SF_{c,\widehat q}[\mu_q]
  =\widehat q\,\delta_{\widehat m+c(y_+-\widehat m)}
  +(1-\widehat q)\,\delta_{\widehat m+c(y_--\widehat m)}.
  \label{eq:composed}
\end{equation}
These are low-dimensional transfer functions to be estimated from controlled
inputs and outputs, not asserted universals.

\section{Behavioral sufficiency on fixed response cells}
\label{sec:twobloc}

We focus on common attraction because centralizing compression naturally
enters this branch and there acts as a shared anchor; uniformly indifferent
or repulsive representations do not provide the same stabilizing channel.
Fixing the AIR labels turns direct interaction into a signed network, while
the represented channel determines which summary statistics remain visible.
The result below holds until a response threshold or opinion boundary is
reached.

\subsection{Exact \texorpdfstring{$K$}{K}-bloc response-cell reduction}

Let
$\mu_x=\sum_{i=1}^Kq_i\delta_{x_i}$ with $q_i>0$ and
$\sum_iq_i=1$. An \emph{open response cell} is a connected set of
configurations on which no direct distance equals either threshold and every
pair keeps the same response label. Set $a_{ii}=0$ by convention (the
diagonal is immaterial), and for $i\ne j$ let
\begin{equation}
  a_{ij}=
  \begin{cases}
    1,&|x_j-x_i|<\epsilon_1,\\
    0,&\epsilon_1<|x_j-x_i|<\epsilon_2,\\
    -\eta,&|x_j-x_i|>\epsilon_2.
  \end{cases}
  \label{eq:response-labels}
\end{equation}
Define the weighted signed Laplacian
\begin{equation}
  (L_Av)_i=\sum_{j=1}^Kq_j a_{ij}(v_i-v_j),
  \label{eq:signed-laplacian}
\end{equation}
the true mean $m=q^\top x$, and the centered configuration
$y=x-m\mathbf1$.

\begin{theorem}[Exact response-cell reduction]
\label{thm:spectral}
Suppose the trajectory remains in the interior of $X^K$ and in one open
response cell. Assume also that every represented position attracts every
recipient,
\begin{equation}
  \sup_{z\in\operatorname{supp}\SF[\mu_x]}|z-x_i|<\epsilon_1
  \quad (i=1,\ldots,K),
  \label{eq:uniform-summary-attraction}
\end{equation}
and write $\widehat m(x)=\int z\,\SF[\mu_x](\dd z)$. Then, until the first
cell or boundary exit,
\begin{equation}
  \boxed{\dot m=\lambda(\widehat m-m)},
  \qquad
  \boxed{\dot y=-[(1-\lambda)L_A+\lambda I]y}.
  \label{eq:spectral-decomposition}
\end{equation}
The operator $L_A$ is self-adjoint in
$\langle u,v\rangle_q=\sum_iq_i u_iv_i$. Let $\kappa_{\min}$ be its smallest
eigenvalue on $H_q=\{v:q^\top v=0\}$. If $\kappa_{\min}<0$, the centered
system changes stability at
\begin{equation}
  \boxed{\lambda_*=
  \frac{-\kappa_{\min}}{1-\kappa_{\min}}}.
  \label{eq:spectral-threshold}
\end{equation}
It has an expanding mode for $\lambda<\lambda_*$, a neutral eigenspace for
$\lambda=\lambda_*$, and is exponentially stable for
$\lambda>\lambda_*$. If $\kappa_{\min}\ge0$, it is exponentially stable for
every $\lambda>0$.
\end{theorem}

\begin{proof}
Inside the cell, direct response satisfies
$V_F[\mu_x](x_i)=-(L_Ax)_i$. Condition
\eqref{eq:uniform-summary-attraction} gives
$V_F[\SF[\mu_x]](x_i)=\widehat m-x_i$, independently of the represented
measure's shape. Symmetry of $a_{ij}$ yields $q^\top L_A=0$, which gives the
mean equation. Subtracting $\dot m\mathbf1$ from the vector equation gives
the centered equation. Moreover,
\begin{equation*}
  \langle u,L_Av\rangle_q
  =\frac12\sum_{i,j}q_iq_j a_{ij}
  (u_i-u_j)(v_i-v_j),
\end{equation*}
so $L_A$ is self-adjoint. A disagreement eigenmode with eigenvalue $\kappa$
has growth rate $\rho(\kappa,\lambda)=-(1-\lambda)\kappa-\lambda$.
The largest rate corresponds to $\kappa_{\min}$ and crosses zero at
Eq.~\eqref{eq:spectral-threshold}.
\end{proof}

Signed-Laplacian spectra and homogeneous anchoring are established tools for
attractive--repulsive networks \citep{bronski2014}. Indeed, the centered
operator in \eqref{eq:spectral-decomposition} is algebraically identical to
the homogeneous transient operator of a continuous-time signed
Friedkin--Johnsen model with equal stubbornness coefficients
\citep{taylor1968,friedkin1990,razaq2025}. The models differ in the source of
that anchor: fixed individual prejudices drive Friedkin--Johnsen dynamics,
whereas the common anchor here is recomputed from the population. Accordingly,
the eigenvalue shift is not claimed as new. Theorem~\ref{thm:spectral}
establishes the preceding model reduction: an endogenous distribution-valued
channel becomes shape-independent in disagreement coordinates, with its mean
retained only in translation.

The threshold is therefore geometric rather than universal. In the three-bloc
example used below, one attractive, one indifferent, and one repulsive pair
give disagreement eigenvalues $-0.221576$ and $0.541576$, hence
$\lambda_*=0.181385$ rather than the complete-repulsion value $2/7$. The
mean-collapse summary satisfies the additional common-attraction condition;
an arbitrary summary need not.

\begin{corollary}[Marginal-mode bias projection]
\label{cor:critical-projection}
Under the hypotheses of Theorem~\ref{thm:spectral}, suppose
$\kappa_{\min}<0$ and the represented mean is
$\widehat m=\widehat q^\top x$ for fixed normalized weights
$\widehat q_i\ge0$ with $\sum_i\widehat q_i=1$, and put
$d=\widehat q-q$. At $\lambda=\lambda_*$, let $P_*$ be the weighted
orthogonal projection onto the $\kappa_{\min}$ eigenspace and let $P_\ell$
project onto each remaining disagreement eigenspace, whose growth rates
$\rho_\ell$ are negative. Then
\begin{equation}
  m(t)=m_0+\lambda_*t\,d^\top P_*y_0
  +\lambda_*\sum_{\ell:\rho_\ell<0}
  \frac{e^{\rho_\ell t}-1}{\rho_\ell}\,d^\top P_\ell y_0
  \label{eq:critical-projection}
\end{equation}
until the first cell or boundary exit. Secular drift occurs if and only if
$d^\top P_*y_0\ne0$.
\end{corollary}

\begin{proof}
Since $\mathbf1^\top d=0$, the mean equation is
$\dot m=\lambda_*d^\top y$. Expand the centered solution of
Eq.~\eqref{eq:spectral-decomposition} in the weighted orthogonal
eigenprojectors and integrate each exponential mode.
\end{proof}

\subsection{Two-bloc specialization: mean-faithful reversal}

The two-bloc case makes the physical content of the reduction transparent. It
is a specialization, not a distinct spectral mechanism: a representation can
preserve the statistic sufficient for linear response while moving the system
to a different AIR branch.

For two equal blocs $\mu_y=\tfrac12\delta_{-y}+\tfrac12\delta_y$ and the
contraction summary $\SF_c$, symmetry preserves the two-bloc form and the
positive bloc obeys
\begin{align}
\dot y={}&\frac{1-\lambda}{2}F(-2y)
+\frac{\lambda}{2}
\left\{F[-(1+c)y]+F[-(1-c)y]\right\}.
\label{eq:twobloc-general}
\end{align}

\begin{corollary}[Two-bloc response-cell reversal]
\label{prop:reversal}
If $2y>\epsilon_2$ and $(1+c)y<\epsilon_1$ (the \emph{mixed-response
wedge}, nonempty iff $c<2\epsilon_1/\epsilon_2-1$), then
\begin{equation}
  \dot y=\left[(1-\lambda)\eta-\lambda\right]y,
  \qquad
  \boxed{\lambda_c=\frac{\eta}{1+\eta}}.
  \label{eq:critical}
\end{equation}
\end{corollary}

\begin{proof}
Direct exposure sees the opposite bloc in the repulsive zone:
$\tfrac{1-\lambda}{2}F(-2y)=(1-\lambda)\eta y$. Both contracted support
points lie in the attractive zone, contributing mean force $-\lambda y$.
\end{proof}

Throughout the trajectory $m(\SF_c[\mu_y])=m(\mu_y)=0$: the summary is
exactly unbiased in the statistic sufficient for linear dynamics, yet it
hides the cross-bloc distance that activates repulsion and substitutes
attraction toward the center. Corollary~\ref{prop:reversal} is a
\emph{local} instability-reversal statement: within the wedge the state
grows exponentially below $\lambda_c$ and contracts above it. Once the
trajectory leaves the wedge, other branches of
Eq.~\eqref{eq:twobloc-general} activate; the eventual destinations observed
numerically---boundary polarization below $\lambda_c$ and consensus above
it, with a $c$-dependent outcome boundary---are properties of the selected
global branches (Fig.~\ref{fig:recursive}), not consequences of the
proposition, and a complete branch enumeration with a solution concept for
the discontinuous measure flow remains open.

\begin{figure*}[t]
  \centering
  \includegraphics[width=0.80\textwidth]{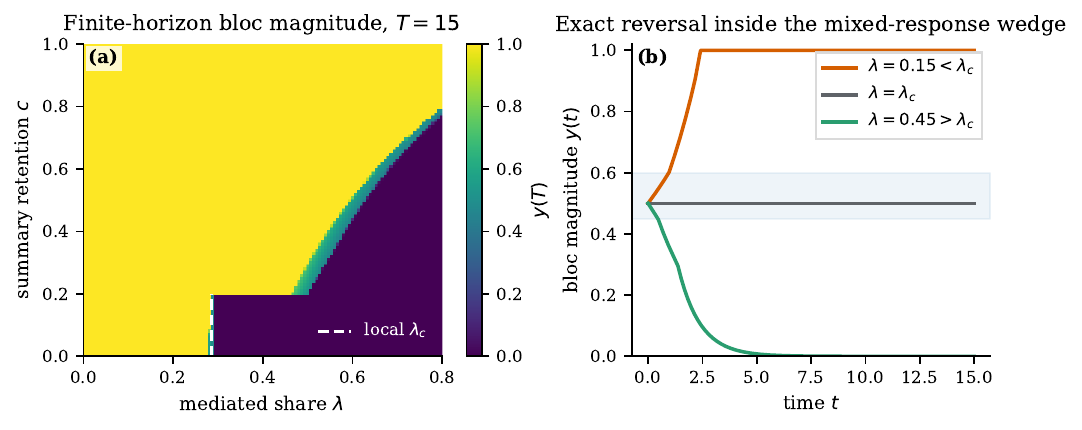}
  \caption{Projected symmetric two-bloc dynamics for $y_0=0.5$ and
  $(\epsilon_1,\epsilon_2,\eta)=(0.6,0.9,0.4)$.
  (a) Bloc magnitude at $T=15$ across $(\lambda,c)$ using projected Euler
  step $0.005$; the dashed segment marks the response-cell threshold.
  (b) Reduced trajectories at $c=0$ below, at, and above
  $\lambda_c=2/7$.}
  \label{fig:recursive}
\end{figure*}

A central attractor competing with threshold-limited peer influence also
arises when a bounded-confidence population receives a \emph{fixed} moderate
broadcast \citep{mckeown2006,vazmartins2010,pansanella2023}, and cluster
drift toward fixed media positions in the Deffuant--Weisbuch model has
recently been analyzed in detail \citep{zheng2026}. The distinction is that
here the attractor is computed from the discussion by a mean-preserving
operator: the effect survives any audit that checks the statistic
sufficient for linear dynamics, and the drift law below is not migration
toward an exogenous position but a gap-invariant translation driven by the
operator's error in representing the evolving prevalence---a mechanism with
no fixed-broadcast analogue.

\subsection{Two-bloc specialization: prevalence-induced translation}

The operators \eqref{eq:false-balance}--\eqref{eq:majority} act on
prevalence, which the symmetric state cannot probe. Let
$\mu=q\delta_{y_+}+(1-q)\delta_{y_-}$ with gap $\Delta=y_+-y_->0$ and true
mean $m=qy_++(1-q)y_-$, summarized by $\SF_{c,\widehat q}$ of
Eq.~\eqref{eq:composed}.

\begin{proposition}[Gap invariance and mean drift]
\label{prop:drift}
Suppose the mixed wedge conditions hold:
\begin{equation}
  \Delta>\epsilon_2,
  \quad
  \big[\max(\widehat q,1-\widehat q)
  +c\min(\widehat q,1-\widehat q)\big]\Delta<\epsilon_1.
  \label{eq:asym-wedge}
\end{equation}
Then, for all $q,\widehat q\in(0,1)$ and $c\in[0,1]$,
\begin{equation}
  \dot\Delta=\big[(1-\lambda)\eta-\lambda\big]\Delta,
  \qquad
  \boxed{\dot m=\lambda\,(\widehat q-q)\,\Delta.}
  \label{eq:drift-laws}
\end{equation}
Within the wedge, the reversal threshold $\lambda_c=\eta/(1+\eta)$ is invariant
under prevalence distortion and support contraction, while any
misrepresentation $\widehat q\ne q$ drives a drift of the entire
configuration; at $\lambda=\lambda_c$ the gap freezes, the wedge conditions
persist, and $m(t)=m(0)+\lambda_c(\widehat q-q)\Delta_0\,t$ exactly until a
bloc reaches the boundary.
\end{proposition}

\begin{proof}
In the wedge the direct forces are $(1-q)\eta\Delta$ and $-q\eta\Delta$. All
represented positions lie in the attractive zone of both blocs, so the
summary force on a bloc at $y$ is $\widehat m-y$, and contraction about
$\widehat m$ preserves $\widehat m$ for every $c$. Hence
$\dot y_+=(1-\lambda)(1-q)\eta\Delta-\lambda(1-\widehat q)\Delta$ and
$\dot y_-=-(1-\lambda)q\eta\Delta+\lambda\widehat q\Delta$. The difference
cancels all $(q,\widehat q,c)$ dependence; the weighted sum
$q\dot y_++(1-q)\dot y_-$ cancels the direct terms and leaves
$\lambda(\widehat q-q)\Delta$. The wedge conditions depend only on
$(\Delta,\widehat q,c)$, all constant on the frozen-gap trajectory.
\end{proof}

For $q>\tfrac12$, Eq.~\eqref{eq:drift-laws} immediately separates the two
summary formats. False balance gives
$\dot m=-\lambda(1-c_q)(q-\tfrac12)\Delta<0$, a drift toward the
\emph{minority}; majority amplification gives $\widehat q>q$ and a drift
toward the \emph{majority}. Neither format moves $\lambda_c$.

In coordinates $(\Delta,m)$ the wedge dynamics are triangular with rates
$\sigma_\Delta=(1-\lambda)\eta-\lambda$ and $\sigma_m=0$: the disagreement
mode is relevant or irrelevant depending on the mediated share, while the
translation mode is marginal and \emph{driven} at a rate proportional to the
operator's prevalence error. The drift law is an exact identity within the
fixed-cell model. In a controlled experiment, with $\lambda$ and
the bloc gap known and other common drift sources excluded, a measured
collective mean drift provides a model-based estimator of the prevalence
error $\widehat q-q$. Identification from observational field data requires
additional assumptions or instruments.

\subsection{Correlation architecture of stochastic summaries}
\label{sec:stochastic}

To isolate error in represented location from error in represented shape, we
model each regeneration error as a translation; cross-recipient correlation
can then vary while the one-output marginal error model is held fixed. The
common-noise principle is standard; the calculation below identifies its
coefficient and validity conditions for the endogenous representation channel.
Let a new stance error
be drawn every $h$ units of time and held until the next regeneration, while
the base operator continues to track the evolving discussion. During refresh
interval $k$, its output is translated by an independent, mean-zero error
$\xi_k$ with variance $\sigma_s^2$. For the exact calculation below, the
errors are bounded so that all summary--recipient displacements remain
strictly in the attraction zone and projection is inactive. Gaussian errors,
which can violate these assumptions and can translate the represented support
outside $X$, are used only as a separately labeled stress test.

\begin{proposition}[Deployment dichotomy]
\label{prop:deployment}
Consider a horizon $T$ that is a multiple of $h$. Condition on the initial
empirical state, take errors to be independent across refresh intervals and,
in the personalized case, across recipients, and assume throughout that every
represented position remains in the recipient's attraction zone and
projection is inactive.

(i) \emph{Shared deployment} (one summary per refresh, broadcast to all $N$
agents). For two point blocs in the mixed wedge, shared translation leaves the
gap equation unchanged and the population mean obeys
\begin{equation}
  \dot m=\lambda\big[(\widehat q-q)\Delta_t+\xi(t)\big],
  \label{eq:common-noise}
\end{equation}
where $\xi(t)$ is piecewise constant. Because $\Delta_t$ is independent of
the common translation error,
$\operatorname{Var}[m(T)-m_{\rm det}(T)\mid x(0)]
=\lambda^2\sigma_s^2hT$, independent of $N$.

(ii) \emph{Personalized deployment} (each agent independently samples its
own displacement $\xi_k^{(i)}$). If the base summary operator preserves the
current empirical mean, then, even though personalized errors broaden the
population,
\begin{equation}
  \dot m_N=\lambda\overline\xi(t),
  \qquad
  \operatorname{Var}[m_N(T)-m_N(0)\mid x(0)]
  =\frac{\lambda^2\sigma_s^2hT}{N}.
  \label{eq:personalized-noise}
\end{equation}
For a prevalence-distorting personalized operator the $1/N$ self-averaging
order remains plausible, but Eq.~\eqref{eq:personalized-noise}'s coefficient
is not asserted: random broadening also randomizes the systematic drift.
\end{proposition}

\begin{proof}
On the attraction branch, translating a represented measure by $\xi$ adds
$\lambda\xi$ to the recipient velocity. A common increment cancels from all
pairwise differences, so the two-bloc gap remains deterministic. For a
mean-preserving personalized operator, oddness cancels the direct pair forces
in the empirical mean and the centered attractive forces also sum to zero,
leaving only $\lambda\overline\xi$. Each refresh interval therefore contributes
$\lambda h\xi_k$ or $\lambda h\overline\xi_k$ to the corresponding mean.
Independence across the $T/h$ intervals gives the two variance formulas.
\end{proof}

For a general shared error process with covariance $C_\xi(t,s)$,
Eq.~\eqref{eq:common-noise} gives variance
$\lambda^2\int_0^T\!\!\int_0^T C_\xi(t,s)\,\dd s\,\dd t$; a mean-preserving
personalized operator gives $1/N$ of it. The phase-aligned block process used
here yields the displayed $\lambda^2\sigma_s^2hT$, while its stationary
random-phase version gives $\lambda^2\sigma_s^2(hT-h^2/3)$ for $T\ge h$.
If instead the entire represented measure is frozen between refreshes, the
consensus-state gain per refresh is $(1-e^{-\lambda h})\xi_k$, with Euler gain
$1-(1-\lambda\delta t)^{h/\delta t}$. The numerical controls distinguish
these refresh semantics.

The dichotomy suggests a mean-field interpretation: a shared sequence should
produce a stochastic conditional measure flow, whereas personalized errors
should average into a deterministic barycentric operator. This is the
common-noise structure of conditional McKean--Vlasov dynamics
\citep{carmona2018}, and the same
common-versus-independent sampling contrast studied for binary opinions with
influencer noise by \citet{coculescu2024}. We have not proved either
mean-field limit for the discontinuous endogenous operator, and do not use
that interpretation as a theorem.

Proposition~\ref{prop:deployment} excludes errors beyond the attraction-cell
tolerance. At $\lambda=\lambda_c$, such exceedances can switch one bloc's
summary response to another branch and perturb the otherwise neutral gap. We
therefore record the first exit from
$[\epsilon_2,2\epsilon_1]$. In the symmetric parameterization used below,
branch enumeration makes these kicks one-sided, so the observed event is an
upper first exit rather than a demonstrated eventual destination. A bounded
variance-matched control never exceeds the tolerance and consequently has no
exit; Gaussian conditions differ through their exceedance probabilities, not
through variance alone. We report paired temporal refinement of the survival
curves and make no post-exit polarization claim.

All conclusions in this section exploit a common response branch. Their
boundary of validity motivates the inverse question: once represented mass
crosses a threshold, which distributional changes can the population detect?
The next section answers that question using the same force observable.

\section{Beyond response cells: identifiability and fidelity}
\label{sec:sufficiency}

The fixed-cell theorem gives a local sufficiency statement: while every
represented stance remains on the same linear attraction branch, represented
shape is invisible to disagreement. We now ask what the same recipients can
distinguish when their probes encounter AIR thresholds. The change is sharp.
Thresholds cast translated ``shadows'' of support points into the observable
force field, so the response law that admits moment reduction within a cell
can identify the underlying empirical measure across cells. The inverse
results below therefore delimit the exact reduction of
Section~\ref{sec:twobloc}; they are not a separate reconstruction model.

\subsection{From moment sufficiency to bounded probes}

The response law fixes the sufficient statistics. For affine
$F(u)=au+b$ with $a\ne0$, equality $V_F[\mu]=V_F[Q]$ at one recipient
position is equivalent to equality of the represented means and therefore
holds at every position. More generally, if $F$ is a polynomial of exact
degree $r\ge1$, equality of the two force fields on a nonempty open interval
is equivalent to equality of the first $r$ moments. This follows directly by
expanding $(z-x)^k$: the force difference is a polynomial in $x$ whose
coefficients are triangular combinations of the moment differences. These
elementary cases provide the reference point for the thresholded law.

At the opposite extreme, if fidelity is demanded for \emph{every} probe on
the real line, classical Fourier analysis applies: for nonzero
$F\in L^1(\mathbb R)$ and compactly supported $\mu,Q$, the condition
$V_F[\mu]=V_F[Q]$ on all of $\mathbb R$ reads $G*(\mu-Q)=0$ with
$G(u)=F(-u)$, and entire-function continuation of the Fourier--Stieltjes
transform forces $\mu=Q$ \citep{folland1999}; this is the single-kernel
analogue of the Fourier-support conditions that make translation-invariant
kernels measure-determining \citep{sriperumbudur2010}. We do not claim this
as new, and it does not settle the physical question: recipients occupy only
$X$, and the law \eqref{eq:air} is not integrable on $\mathbb R$ as written.

\subsection{Threshold shadows identify finite discussions}

The physically relevant question is whether probes confined to the opinion
domain identify the discussion. For the class of measures that discussions
actually generate---finite empirical measures---the answer is yes under
explicit visibility and noncollision conditions.

\begin{theorem}[Restricted-probe identifiability]
\label{thm:restricted}
Let $F$ be the law \eqref{eq:air} with $\epsilon_1<1$. Let
$\mu\in\PP(X)$ be finitely supported and let $Q\in\PP(X)$ be arbitrary, with
atom set $A(Q)$. Call the combined atomic support
$S=\operatorname{supp}\mu\cup A(Q)$ \emph{generic} if no two distinct points
$a,b\in S$ satisfy $|a-b|$ in
$\{2\epsilon_1,\,2\epsilon_2,\,\epsilon_2-\epsilon_1,\,
\epsilon_1+\epsilon_2\}$ and no point $a\in S$ satisfies
$a-\tau\in\{-1,1\}$ for some $\tau\in T$. If $S$ is generic and
$V_F[\mu](x)=V_F[Q](x)$ for all $x\in(-1,1)$, then $Q=\mu$.
\end{theorem}

\begin{proof}
Write $Q=Q_a+Q_c$ with $Q_a$ atomic and $Q_c$ atomless, let $\nu=\mu-Q$ and
$W=V_F[\nu]$. The atomless part contributes a \emph{continuous} function of
$x$: $F$ is bounded with finitely many discontinuity points, so as
$x\to x_0$ the integrand $F(z-x)$ converges to $F(z-x_0)$ for all $z$
outside the $Q_c$-null set $\{x_0+\tau:\tau\in T\}$, and dominated
convergence applies. For an atom at $a$, the map $x\mapsto F(a-x)$ is
continuous except at the four shadow points $x=a-\tau$, $\tau\in T$, where
its right-minus-left jump is $-J_\tau$ because increasing $x$ traverses the
threshold in the reverse direction (magnitudes $\epsilon_1$ at
$\pm\epsilon_1$, $\eta\epsilon_2$ at $\pm\epsilon_2$); the atom series
converges uniformly (Weierstrass, $\sum_a w_a\sup|F|<\infty$), so one-sided
limits pass through the sum and the jump of $W$ at $x_0$ is
$-\sum_{\tau\in T}J_\tau[\mu-Q_a](\{x_0+\tau\})$. Two distinct atoms
$a\ne b$ contribute to the same shadow point only if $a-b=\tau-\tau'$, and
the set of such differences is exactly
$\pm\{2\epsilon_1,2\epsilon_2,\epsilon_2-\epsilon_1,
\epsilon_1+\epsilon_2\}$, excluded by genericity. Each shadow point in
$(-1,1)$ therefore isolates one atom through one threshold, and $W\equiv0$
forces $\mu(\{a\})=Q_a(\{a\})$ for every $a\in S$ with a shadow in
$(-1,1)$. Coverage holds for every $a\in[-1,1]$: if $a>\epsilon_1-1$ then
$a-\epsilon_1\in(-1,1)$, and otherwise $a+\epsilon_1\le2\epsilon_1-1<1$;
the boundary clause excludes shadows landing exactly on $\pm1$. Hence the
atomic parts coincide, so $Q_a$ has total mass $1$ and $Q_c=0$.
\end{proof}

\begin{theorem}[Finite-support inverse stability]
\label{thm:inverse-stability}
Let $F$ be the law \eqref{eq:air} with $\epsilon_1<1$, let
$\mu,Q\in\PP(X)$ be finitely supported, let their combined support
$S$ satisfy the genericity conditions of Theorem~\ref{thm:restricted}, and
write $s=|S|$ and $j_*=\min\{\epsilon_1,\eta\epsilon_2\}>0$. We use the
total-variation distance
$\|\mu-Q\|_{\mathrm{TV}}=\sup_{B\in\mathcal B(X)}|\mu(B)-Q(B)|$. Then
\begin{equation}
  \|\mu-Q\|_{\mathrm{TV}}
  \le \frac{s}{j_*}\DF(\mu,Q),
  \qquad
  \Wone(\mu,Q)
  \le \frac{2s}{j_*}\DF(\mu,Q).
  \label{eq:inverse-stability}
\end{equation}
If $\mu=\sum_{i=1}^n p_i\delta_{a_i}$ and a generic $Q$ has fewer than
$n$ support points, then
\begin{equation}
  \DF(\mu,Q)\ge\frac{j_*}{2}\min_i p_i.
  \label{eq:compression-floor}
\end{equation}
Thus cardinality reduction has a nonzero behavioral distortion floor on each
generic finite configuration.
\end{theorem}

\begin{proof}
Put $W=V_F[Q]-V_F[\mu]$. For each $a\in S$, the coverage argument in
Theorem~\ref{thm:restricted} supplies a shadow $x=a-\tau\in(-1,1)$ that no
other support point shares. The jump of $W$ there is
$-J_\tau[Q-\mu](\{a\})$. Both one-sided limits have magnitude at most $\DF$,
and hence
\begin{equation*}
  |[Q-\mu](\{a\})|\le \frac{2\DF}{|J_\tau|}
  \le\frac{2\DF}{j_*}.
\end{equation*}
Summing over $S$ and using
$\|\mu-Q\|_{\mathrm{TV}}=\tfrac12\sum_{a\in S}
|[\mu-Q](\{a\})|$ proves the first inequality. The second follows from
$\Wone\le\operatorname{diam}(X)\|\mu-Q\|_{\mathrm{TV}}=2\|\mu-Q\|_{\mathrm{TV}}$.
If $Q$ has fewer support points than $\mu$, some $a_i$ is absent from its
support; the isolated shadow at that atom has jump at least $j_*p_i$, whereas
any function bounded by $\DF$ has jumps no larger than $2\DF$.
\end{proof}

These are structural identifiability statements for the behavioral operator,
not finite-sample reconstruction guarantees. The probe variable $x$ indexes
possible recipient states, and the uniform force norm is a worst-case
dynamical discrepancy. Finite-agent discussions are atomic by construction
and need not have small support; the substantive idealizations are exact
response discontinuities and continuum probing. Conditioning under finite
probe grids, observation noise, or smooth transition layers remains open.

The identification theorem places no structural restriction on the summary
output: within
generic position of the atomic supports, \emph{no representation that alters
a finitely supported discussion at all---whether atomic, continuous, or
mixed---is exactly dynamically faithful}, even though recipients never probe
outside the opinion domain. Its quantitative inverse bound additionally
requires $Q$ to be finitely supported. For fixed finite atomic support
cardinalities, the excluded configurations form a finite union of affine
hyperplanes in the support-location parameter space; no codimension claim is
made for arbitrary probability measures.

\begin{remark}
\label{rem:inverse}
Theorem~\ref{thm:restricted} specializes the general inverse problem of
identifying atomic measures from kernel responses---the setting of
finite-rate-of-innovation sampling and Prony-type reconstruction
\citep{vetterli2002}---to the restricted-probe geometry created by
thresholded social influence, where the ``sampling kernel'' is a behavioral
law and its jump set does the identification. Two directions remain open: the
non-generic lattice case, where jump
contributions can cancel and the constraint
$\sum_{\tau\in T}J_\tau\,\nu(\{x_0+\tau\})=0$ governs possible nullspaces;
and discussions that are themselves non-atomic, where the analogous object
is the kink calculus of $W$. Theorem~\ref{thm:inverse-stability} resolves
stable recovery for generic finite supports in the uniform force norm. If two
finitely supported candidates both fit an observed force field within uniform
error $\zeta$, and their combined support satisfies the theorem's genericity
conditions and has $m$ points, then their total-variation distance is at most
$\min\{1,2m\zeta/j_*\}$. In particular, if each candidate has at most $s$
support points, the bound is $\min\{1,4s\zeta/j_*\}$. This topology is
necessarily strong: moving an atom across a response jump by an arbitrarily
short distance produces order-one force error. The result complements
stability theory for smooth sampling kernels \citep{candes2014}.
\end{remark}

\subsection{Threshold-sensitive distortion: transport and cumulative mass}
\label{sec:transport}

For $L_F$-Lipschitz response, Kantorovich--Rubinstein duality gives
immediately $\DF(\mu,Q)\le L_F\Wone(\mu,Q)$ \citep{villani2009}: transport
accuracy suffices. Thresholded laws are only piecewise Lipschitz, and the
implication fails in one direction only. Let
$J=\{\tau_1,\ldots,\tau_K\}$ be the jump set, $F$ be $L$-Lipschitz on each
component of $(X-X)\setminus J$, $M=\sup|F|$, and define the threshold-layer
mass
\begin{equation}
  M_\delta(\mu)
  =\sup_{x\in X}\mu\{z:\operatorname{dist}(z-x,J)\le\delta\}.
  \label{eq:threshold-mass}
\end{equation}

\begin{theorem}[Threshold-aware upper bound]
\label{thm:threshold-bound}
For every $\delta>0$,
\begin{equation}
  \DF(\mu,Q)
  \le L\Wone(\mu,Q)
  +2M\left[
    \frac{\Wone(\mu,Q)}{\delta}+M_\delta(\mu)
  \right].
  \label{eq:threshold-bound}
\end{equation}
If $\mu$ has a density bounded by $\rho$, optimizing $\delta$ gives
$\DF\le Lw+4M\sqrt{2\rho Kw}$ with $w=\Wone(\mu,Q)$.
\end{theorem}

\begin{proof}
Couple $(Z,Z')$ optimally. Pairs whose difference interval crosses no jump
contribute at most $L|Z-Z'|$; a crossing contributes at most $2M$ and, if
$|Z-Z'|\le\delta$, requires $\operatorname{dist}(Z-x,J)\le\delta$; pairs with
$|Z-Z'|>\delta$ have probability at most $\Wone/\delta$ by Markov. Take
expectations and the supremum over $x$. The density case bounds
$M_\delta\le2\rho K\delta$ and sets $\delta=\sqrt{w/(2\rho K)}$.
\end{proof}

Because $M_\delta$ takes a supremum over recipients, the bound is
informative for spread-out discussions and deliberately weak for atomic
ones. The complementary question---when is large distortion
\emph{forced}---is answered by a cumulative-mass statistic, not a transport
one.

\begin{theorem}[Cumulative-discrepancy sandwich]
\label{thm:ks-lower}
Let $F$ have jump sizes $J_k=F(\tau_k^+)-F(\tau_k^-)$ at $\tau_k\in J$, be
$L$-Lipschitz on each component of $(X-X)\setminus J$, and attain one of its
one-sided limits at each jump. Define the \emph{oriented} tail
$A_k(x)=\{z:z-x\ge\tau_k\}$ at thresholds where $F$ is right-continuous and
$A_k(x)=\{z:z-x>\tau_k\}$ where it is left-continuous---for the law
\eqref{eq:air}, closed tails at $+\epsilon_1,+\epsilon_2$ and open tails at
$-\epsilon_1,-\epsilon_2$---and the threshold-offset discrepancy
\begin{equation}
  S_J^{Y}(\mu,Q)
  =\sup_{x\in Y}
  \Big|\sum_{k}J_k\,[Q-\mu]\big(A_k(x)\big)\Big|.
  \label{eq:sj}
\end{equation}
Then for every nonempty probe set $Y\subseteq X$,
\begin{equation}
  \big|\DF^{Y}(\mu,Q)-S_J^{Y}(\mu,Q)\big|
  \le L\,\Wone(\mu,Q).
  \label{eq:ks-lower}
\end{equation}
If a single threshold $\tau$ is active on the combined support,
$S_J^Y=|J_\tau|\sup_{x\in Y}|[Q-\mu](A_\tau(x))|$: the distortion equals a
probe-restricted Kolmogorov--Smirnov-type statistic at threshold offsets, up
to an error no larger than $L\Wone$.
\end{theorem}

\begin{proof}
Write $F(u)=G(u)+\sum_k J_k H_k(u)$ with $H_k(u)=\mathbf 1[u\ge\tau_k]$ at
right-continuous jumps and $\mathbf 1[u>\tau_k]$ at left-continuous ones.
Each subtraction removes its jump on the side where $F$ attains its limit,
so $G$ is continuous on $X-X$ and $L$-Lipschitz there. For fixed $x$ the
force difference is the oriented jump sum plus a $G$-term bounded in modulus
by $L\Wone$ (Kantorovich--Rubinstein duality), so the two quantities differ
pointwise by at most $L\Wone$. Since
$|\sup_x|a(x)|-\sup_x|b(x)||\le\sup_x|a(x)-b(x)|$, taking suprema over
$x\in Y$ gives \eqref{eq:ks-lower}.
\end{proof}

\begin{remark}
The orientation is not cosmetic. With a uniformly closed convention,
$\mu=\delta_{-\epsilon_1}$ and $Q=\delta_{-\epsilon_1+\varepsilon}$ probed
at $x=0$ give a naive discrepancy of zero while
$\DF^{\{0\}}=\epsilon_1-\varepsilon$, violating \eqref{eq:ks-lower} by an
order-one amount: the law is left-continuous at $-\epsilon_1$, so the atom
sitting exactly on the threshold belongs to the indifferent branch and only
the open tail counts it correctly. The oriented convention matters only for
atoms located exactly at threshold offsets. The numerical verification
includes all four endpoint conventions, including this configuration.
\end{remark}

Theorem~\ref{thm:ks-lower} says that for thresholded response the dynamical
distortion \emph{is} the cumulative discrepancy at threshold offsets, up to
an error controlled by the transport distance; together with
Theorem~\ref{thm:threshold-bound} it identifies the threshold-sensitive
quantities that must be controlled in addition to transport error, and
Theorem~\ref{thm:threshold-bound} remains useful when only $\mu$'s
threshold-layer mass, rather than cumulative discrepancies, is known. For the
single probe $Y=\{0\}$, the atomic pair
$\mu_\varepsilon=\delta_{\tau-\varepsilon}$ and
$Q_\varepsilon=\delta_{\tau+\varepsilon}$ has
$\Wone=2\varepsilon\to0$. If $\varepsilon$ is small enough that no other
threshold lies between the atoms, then
$S_J^{\{0\}}=|J_\tau|$ and the force error is within $2L\varepsilon$ of that
jump. The restriction matters: over $Y=X$, translating the recipient can
activate another threshold, and the supremum need not select $J_\tau$.
Thus transport convergence does not imply dynamical convergence when
cumulative mass crosses a threshold offset. Conversely, if $S_J^Y$ and
$\Wone$ both vanish then $\DF^Y$ vanishes---a statement the one-sided bounds
alone would not justify, because the layer mass $M_\delta(\mu)$ in
Eq.~\eqref{eq:threshold-bound} need not vanish for atomic $\mu$.
Consequently, average displacement metrics cannot certify dynamical fidelity,
and a discontinuity alone does not preclude it---what matters is whether
cumulative discussion mass moves across a threshold offset.

\subsection{Dynamical consequences of force distortion}
\label{sec:recursive}

The preceding results concern a static input--output map. Two consequences
connect that map to evolution. First, for a Lipschitz response, a standard
Dobrushin argument turns one-step distortion into a finite-horizon trajectory
bound. Define the unmediated velocity
$b_0[\nu]=V_F[\nu]$ and the mediated velocity
\begin{equation*}
  b_\lambda[\nu]=V_F[\nu]
  +\lambda\{V_F[\SF[\nu]]-V_F[\nu]\}.
\end{equation*}
Let $F$ be $L_F$-Lipschitz and let $\mu_t$ and
$\widetilde\mu_t$ be projected characteristic flows on $X$ driven by $b_0$
and $b_\lambda$, respectively, with the same initial measure. Assuming these
flows exist, for every $t$ in their common interval of existence,
\begin{equation}
  \Wone(\mu_t,\widetilde\mu_t)
  \le\lambda\int_0^t
  e^{2L_F(t-s)}
  \DF(\widetilde\mu_s,\SF[\widetilde\mu_s])\,\dd s .
  \label{eq:recursive-bound}
\end{equation}
Indeed, transporting a common initial coupling along the two characteristic
flows, the
Lipschitz force contributes at most
$2L_F\mathbb E|X_t-\widetilde X_t|$: one term comes from the recipient
position and one from the measure argument. The summary forcing contributes
at most $\lambda\DF(\widetilde\mu_t,\SF[\widetilde\mu_t])$. The normal cone
of the convex interval is monotone, so projection adds no positive term to
the distance derivative. Gronwall's inequality gives
Eq.~\eqref{eq:recursive-bound}.
This established estimate \citep{dobrushin1979,toscani2006} is included only
to distinguish error production from dynamical amplification. For a mollified
response $F_\varrho$ it holds with $L_{F_\varrho}=O(1/\varrho)$; no analogous
global uniqueness or stability claim is made here for the discontinuous AIR
flow. The strict response-cell results avoid that unresolved continuation
problem.

Second, the behavioral distortion directly bounds a macroscopic observable
without any Lipschitz assumption.

For any odd response law and any empirical state whose agents are in the
interior of $X$, the population mean $m_N=N^{-1}\sum_i x_i$ satisfies
\begin{equation*}
  \begin{aligned}
    \dot m_N
    &=\lambda\int_X\bigl(V_F[\SF[\mu_N]]-V_F[\mu_N]\bigr)(x)
      \,\mu_N(\dd x),\\
    |\dot m_N|&\le\lambda\DF(\mu_N,\SF[\mu_N]).
  \end{aligned}
\end{equation*}
When projection is active, the right-hand side gains the boundary-reaction
term $N^{-1}\sum_i\{\Pi_{T_X(x_i)}B_i-B_i\}$.
Indeed, oddness makes the direct pairwise force vanish after summing over
ordered pairs. The same cancellation gives
$\int V_F[\mu_N]\,\dd\mu_N=0$, and the uniform bound follows from the
definition of $\DF$.

\section{Numerical methods}
\label{sec:methods}

The numerical study follows the theoretical regime map. Static force-field
experiments isolate moment-sufficient and threshold-sensitive compression;
fixed-cell integrations resolve translation and disagreement; projected
particles test the reduced descriptions beyond ideal point blocs; and
finite-size ensembles expose the role of error correlation. All experiments
use $\epsilon_1=0.60$, $\epsilon_2=0.90$, and
$\eta=0.40$. These dimensionless values are not fitted to behavioral data:
they make all three AIR branches accessible. At the reference state $y_0=0.5$,
direct exposure is repulsive ($2y_0=1>\epsilon_2$) while mean collapse is
attractive ($y_0<\epsilon_1$), away from switching surfaces. The choice
$\eta=0.40$ gives the interior threshold
$\lambda_c=2/7\approx0.2857$. The analysis depends on these inequalities,
not on the particular values. The first-exit calculations deliberately cross
those cells and are reported as stress tests rather than theorem validation.
Every pseudorandom calculation uses an explicit seed.

\subsection{Force fidelity across response thresholds}

The static fidelity experiment uses the symmetric discussion
$\mu=(\delta_{-0.5}+\delta_{0.5})/2$ and its mean-collapse representation
$\SF_0[\mu]=\delta_0$. The force fields generated by the two measures are
evaluated at 2001 equally spaced recipient positions in $[-1,1]$, first under
linear influence and then under the thresholded law \eqref{eq:air}. This
comparison holds the represented mean fixed while changing only the support,
thereby separating moment fidelity from fidelity of the nonlinear force
field.

Threshold sensitivity is examined through two perturbation families with
the same transport scale. In the atomic family, unit mass is moved from
$\epsilon_2-\varepsilon$ to $\epsilon_2+\varepsilon$, with 34 logarithmically
spaced values $10^{-5}\leq\varepsilon\leq5\times10^{-3}$. The corresponding
force error at $x=0$ is compared with the jump limit $\eta\epsilon_2$. As a
continuous control, a truncated Gaussian density on $[0.70,0.98]$, centered
at $0.88$ with standard deviation $0.04$, is translated by
$2\varepsilon$. Its force is evaluated by quadrature on 280001 points. The
atomic and continuous constructions therefore have identical $W_1$ distance
but sharply different threshold-layer mass.

The analytical distortion bounds are also evaluated on 200 independently
generated pairs of atomic probability measures, each containing between one
and five support points. Distortion is computed on 801 recipient positions.
The random-measure calculation uses the oriented half-line convention of
Theorem~\ref{thm:ks-lower}; separate adversarial cases place atoms exactly at
negative-threshold offsets to test the endpoint convention. These evaluations
test the numerical implementation of the bounds; the validity of the
inequalities follows from their proofs.

\subsection{Response-cell modes and finite-population dynamics}

The symmetric reduction \eqref{eq:twobloc-general} is computed to $T=15$ on
a $121\times101$ grid in mediated share
$\lambda\in[0,0.8]$ and retention $c\in[0,1]$, starting from $y_0=0.5$.
The finite-horizon surface uses projected Euler steps of $0.005$; the
representative trajectories use projected Runge--Kutta updates at the same
nominal step. Because the vector field is discontinuous, no fourth-order
convergence claim is made. Outcome classifications are compared on a nested
grid against step $0.010$ and horizon $T=30$. The map is a finite-horizon
diagnostic, not an asymptotic phase diagram.

Finite-population validation uses $N=240$ globally coupled agents governed
by Eq.~\eqref{eq:particle}. For the symmetric experiment, the agents form two
mirrored narrow blocs around $\pm0.5$ and are evolved to $T=12$ with time
step $0.02$. The observable is mean extremism
$N^{-1}\sum_i|x_i(t)|$, compared at each recorded time with the two-bloc
solution.

The signed-Laplacian spectrum is first checked at the asymmetric state
\begin{equation*}
  x=(-0.5,-0.3,0.5),\qquad q=(0.3,0.3,0.4),
\end{equation*}
where the two disagreement eigenvalues and the resulting threshold are the
values reported after Theorem~\ref{thm:spectral}. The plotted experiment scans
$x_r=(-r,0,r)$ over 421 radii $r\in[0.16,0.98]$, using
\begin{equation*}
  q^{\rm central}=(0.30,0.40,0.30),\qquad
  q^{\rm outer}=(0.40,0.20,0.40).
\end{equation*}
Radii within $2.5\times10^{-3}$ of a response threshold are omitted because
no open response cell is defined there. The prevalence experiment starts
from blocs at $y_+=0.47$ and $y_-=-0.47$
with true positive prevalence $q=0.55$. Prevalence-faithful, false-balance
($c_q=0.4$, giving $\widehat q=0.52$),
and majority-amplifying operators are compared at the marginal mediated
share. The mean trajectory is evaluated against
$m(0)+\lambda(\widehat q-q)\Delta_0t$ while the cell remains fixed.
The corresponding particle experiment uses 132 positive and 108 negative
agents and a recursively updated two-cluster summary with amplification
$a=2.2$. The deterministic particle trajectory is compared with the
asymmetric two-bloc reduction; robustness to independent opinion noise of
intensity $0.015$ is summarized by 128 realizations and a pointwise central
95\% empirical trajectory band.

A separate mode-resolved experiment uses
$x=(-0.55,-0.20,0.45)$ with $q=(0.30,0.30,0.40)$, which has the same
one-attractive, one-indifferent, one-repulsive response geometry as the
example following Theorem~\ref{thm:spectral}. For
$\lambda=0,0.01,\ldots,0.40$, the fixed-cell equations are integrated to
$T=0.25$ with step $5\times10^{-4}$. Each growth rate is estimated as the
logarithmic change of its weighted modal amplitude and compared with
$\rho_k(\lambda)=-(1-\lambda)\kappa_k-\lambda$.
At $\lambda_*=0.181385$, two prevalence perturbations are then evolved to
$T=20$ with step $0.002$: represented weights
$(0.26,0.30,0.44)$ have nonzero projection on the neutral mode, whereas
$(0.396630,0.18,0.423370)$ are constructed to be orthogonal to it. Their
mean trajectories are compared with Eq.~\eqref{eq:critical-projection}; all
direct distances, represented distances, and positions are monitored to
verify that the assumed response cell and interior conditions persist.

\subsection{Error correlation and marginal gap-band exit}

The deployment experiment compares shared and personalized summary errors at
$\lambda=0.45>\lambda_c$. Errors have standard deviation
$\sigma_s=0.05$ and are refreshed every $h=0.05$ over $T=40$. In the
conditional attraction-cell ensemble, uniform errors of half-width
$\sqrt{3}\sigma_s$ are propagated through the exact attraction-cell mean
recurrence for
$N\in\{30,60,120,240,480,960\}$, with 256 independent realizations per
size and architecture. The recurrence is exact while the represented signal
remains attractive and projection is inactive. Eight full-particle paths per
architecture at $N=30$ spot-check the recurrence over the complete horizon
and record minimum attraction and boundary margins; the recurrence ensembles
remain conditional calculations rather than a global cell-invariance claim.
The unbounded Gaussian
condition remains a nonlinear particle stress test with step $0.01$, 64
realizations, and $N\in\{30,120,480\}$. Percentile bootstrap $95\%$
intervals for variances and log--log slopes use 2000 resamples.

The bounded errors match the Gaussian variance but cannot exceed the initial
wedge tolerance $0.14$; Gaussian errors can, with per-refresh probability
$5.1\times10^{-3}$ for one idealized broadcast at $\sigma_s=0.05$. This
contrast separates conditional attraction-cell scaling from threshold-crossing
stress tests. The fitted dependence of final mean variance on $N$ tests the
known powers in Proposition~\ref{prop:deployment}; it is not presented as a
critical exponent or universality classification.

Two numerical controls distinguish refresh semantics and discretization
effects. The held-error dynamics are compared with a fully frozen summary,
using both the continuous and finite-step predictions stated above. Temporal
refinement uses paired noise streams at fixed $h$ with integration steps
$0.01$ and $0.005$. Finally, first exit of the gap from its marginal band is
studied in the exact two-bloc system at $\lambda=\lambda_c$, starting from
$y_\pm=\pm0.46$. For each of 2000 realizations, the first exit time of the gap
from $[\epsilon_2,2\epsilon_1]=[0.9,1.2]$ is recorded up to $T=200$, together
with the exit direction. Summary-response branches may switch before a gap
exit, so this observable is not called response-cell survival. Gaussian errors with
$\sigma_s\in\{0.04,0.05,0.06\}$ are compared with the variance-matched
uniform distribution. The two timesteps use identical pre-generated
refresh-error arrays; crossings are linearly interpolated within the detecting
step, and convergence is measured by the supremum distance between survival
curves and by paired exit-time and exit-status differences.
Raw samples, refinement records, seeds, software versions, and file hashes are
archived with the aggregate results.

\section{Results}
\label{sec:results}

The results follow the same forward-to-inverse logic as the analysis. We first
test the regime in which the representation enters the collective modes only
through its mean, then examine the two-bloc and stochastic consequences, and
finally cross AIR thresholds to show where moment and transport fidelity can
cease to control behavior. Numerical agreement with a closed formula is used
as a consistency check; the scientific comparisons concern mode selection,
cell exit, finite-population behavior, and correlation architecture.

\subsection{Exact response-cell reduction is mode selective}

Figure~\ref{fig:modes} resolves both disagreement modes of the asymmetric
three-bloc response cell. The direct signed Laplacian has
$\kappa=(-0.221576,0.541576)$. The first mode changes from growth to decay at
$\lambda_*=0.181385$, while the second remains strictly stable throughout
the scan. Fixed-step estimates of $\rho_k$ remain within
$1.4\times10^{-4}$ of $-(1-\lambda)\kappa_k-\lambda$, and every trajectory
remains at least $0.0188$ from a direct-response threshold, $0.0743$ inside
the summary-attraction zone, and $0.4349$ from the opinion boundary. Thus the
local loss-of-stability threshold is not a universal two-bloc value: it is
selected by the most negative mode of the current signed response geometry.

At $\lambda_*$, the active prevalence perturbation has neutral projection
$0.037559$ and predicted secular slope $0.006813$. The secular contribution
over $T=20$ is $0.1363$ and the total translation, including the stable-mode
transient, is $m(20)-m(0)=0.1370$. The perturbation constructed orthogonally
to the neutral mode has zero secular slope and only a decaying transient. Equation
\eqref{eq:critical-projection} matches the integrated means within
$1.3\times10^{-6}$. The smallest margins to a direct-response threshold,
summary-attraction threshold, and opinion boundary are respectively
$0.0390$, $0.0551$, and $0.3995$, so this comparison remains inside the
theorem's stated cell rather than relying on post-switch behavior.

\begin{figure*}[t]
  \centering
  \includegraphics[width=0.80\textwidth]{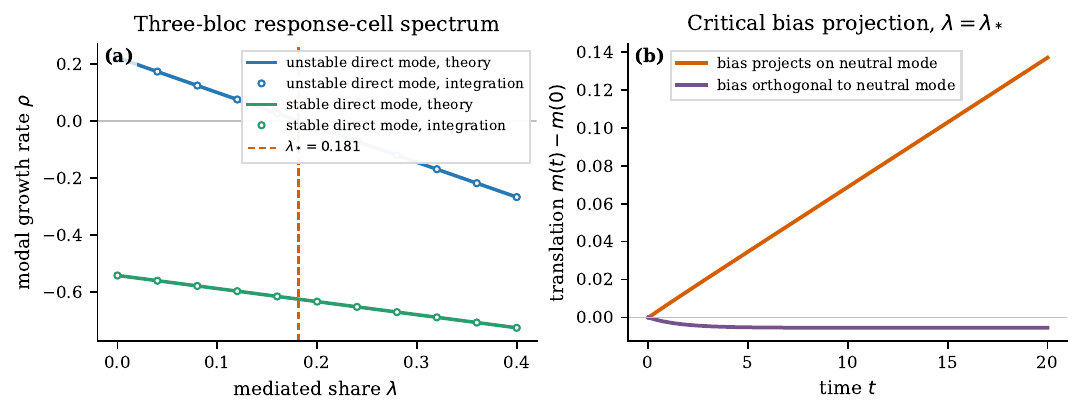}
  \caption{Mode-resolved validation of the three-bloc decomposition for
  $x=(-0.55,-0.20,0.45)$ and $q=(0.30,0.30,0.40)$.
  (a) Theoretical fixed-cell growth rates (lines) and rates estimated from
  numerical integration (symbols); the dashed line marks $\lambda_*$.
  (b) Translation at marginal stability for a prevalence bias with nonzero neutral
  projection and for a nonzero bias orthogonal to that mode. Solid curves are
  integrations and dashed curves are Eq.~\eqref{eq:critical-projection}.}
  \label{fig:modes}
\end{figure*}

\subsection{Two-bloc response-cell reversal at finite horizon}

At $c=0$ the finite-horizon outcome switch lies one grid interval from $2/7$
(Fig.~\ref{fig:recursive}). The displayed low-mediation trajectory reaches
the projected boundary, the high-mediation trajectory reaches consensus, and
the marginal initial state remains stationary while it stays in the response
cell. Increasing retention changes which summary--recipient distances are
attractive, indifferent, or repulsive, producing a $c$-dependent outcome map.
These are numerical states at the declared horizon rather than a proof of an
asymptotic bifurcation boundary. Outcome classifications agree at every point of the
nested grid when the step is reduced from $0.010$ to $0.005$; extending the
horizon from $T=15$ to $T=30$ changes $1.48\%$ of classifications, which
quantifies the coarse outcome sensitivity. The corresponding bloc magnitudes
can still differ by $0.467$ (99th percentile $0.323$), so value-level
convergence is not claimed near the switching regions. The $N=240$
particle system tracks the reduction within $8.1\times10^{-3}$ in mean
extremism (Fig.~\ref{fig:particles}).

\begin{figure*}[t]
  \centering
  \includegraphics[width=\textwidth]{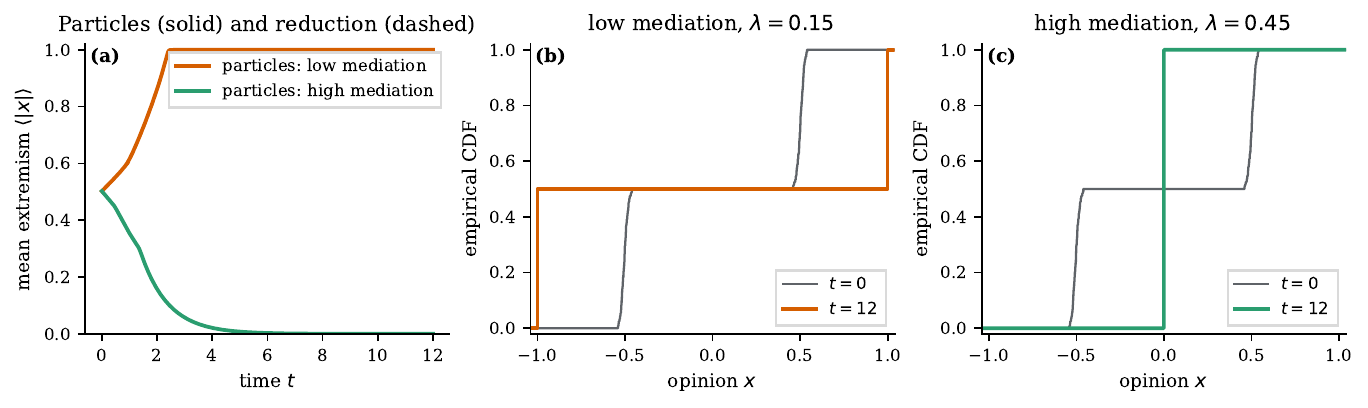}
  \caption{Projected $N=240$ particle validation at $T=12$ and step $0.02$.
  (a) Mean extremism for particles (solid) and the symmetric reduction
  (dashed). (b,c) Initial and final empirical CDFs for
  $\lambda=0.15$ and $0.45$, respectively.}
  \label{fig:particles}
\end{figure*}

\subsection{Prevalence error excites translation}

Figure~\ref{fig:drift} combines a validation of the two-bloc drift law with a
numerical check of the multi-bloc response-cell spectrum. At
$\lambda=\lambda_c$ the gap stays at $\Delta_0=0.94$ to within
$2\times10^{-13}$ for all three operators while the mean moves as predicted:
conserved at $0.047$ under the prevalence-faithful operator, drifting to
$-0.114$ under false balance and to $+0.362$ under majority amplification by
$T=20$. Unequal-mass particles with the recursive cluster summary follow the
reduction with maximum mean error $3.9\times10^{-5}$. The center panel uses a
three-position response geometry to expose changes in the leading
disagreement eigenvalue and the spectral threshold of
Theorem~\ref{thm:spectral}. With opinion noise $0.015$, the 128-run ensemble
retains a positive average drift but exhibits substantial trajectory
dispersion. Every path leaves summary attraction and 125 paths activate
boundary projection. This is a fully nonlinear projected stress test, not a
perturbative validation of the closed-form coefficient.

\begin{figure*}[t]
  \centering
  \includegraphics[width=\textwidth]{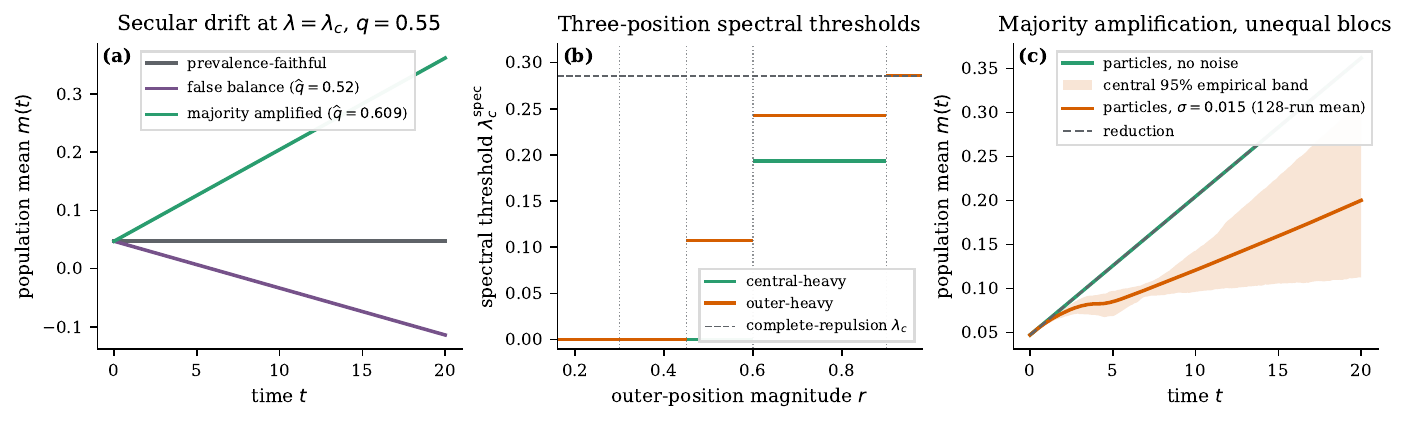}
  \caption{Response-cell modes and prevalence drift for
  $(\epsilon_1,\epsilon_2,\eta)=(0.6,0.9,0.4)$.
  (a) Two-bloc mean trajectories at $\lambda_c$; dotted curves are
  $m(0)+\lambda_c(\widehat q-q)\Delta_0t$.
  (b) Spectral mediation threshold in a three-position response geometry,
  with vertical lines at branch thresholds.
  (c) Unequal-mass particles under a majority-amplifying summary: noiseless
  reduction and particles, plus the mean and pointwise central 95\% empirical
  trajectory band from 128 noisy realizations.}
  \label{fig:drift}
\end{figure*}

\subsection{Error correlation controls collective fluctuations}

Figure~\ref{fig:regimes} illustrates Proposition~\ref{prop:deployment} at
$\lambda=0.45>\lambda_c$, where the population collapses toward consensus.
With bounded attraction-cell errors, bootstrap log--log slopes across six
population sizes are $0.002$ (95\% interval $[-0.061,0.065]$) for shared
errors and $-1.030$ ($[-1.088,-0.969]$) for personalized errors, matching the
analytical powers $0$ and $-1$. The shared variance remains near
$1.0125\times10^{-3}$, whereas personalized variance falls to
$1.03\times10^{-6}$ at $N=960$. Sixteen full-particle spot checks agree with
the exact conditional recurrence to numerical precision and retain positive
attraction and boundary margins. They validate the implementation, not global
cell invariance.

Gaussian errors leave summary attraction more often under personalized
deployment because any recipient can trigger a branch switch: over the full
horizon the exit fractions rise from $0.203$ at $N=30$ to $0.938$ at $N=480$,
versus $0.016$ to $0.063$ for shared errors. The corresponding slopes are
therefore reported only as stress-test descriptors. Paired refinement changes
the final mean by at most $1.3\times10^{-6}$. The frozen-summary control is
consistent with both its finite-step and continuous-time variance predictions.
Thus identical per-output error statistics do not determine the collective
regime; cross-recipient correlation does.

\begin{figure*}[t]
  \centering
  \includegraphics[width=\textwidth]{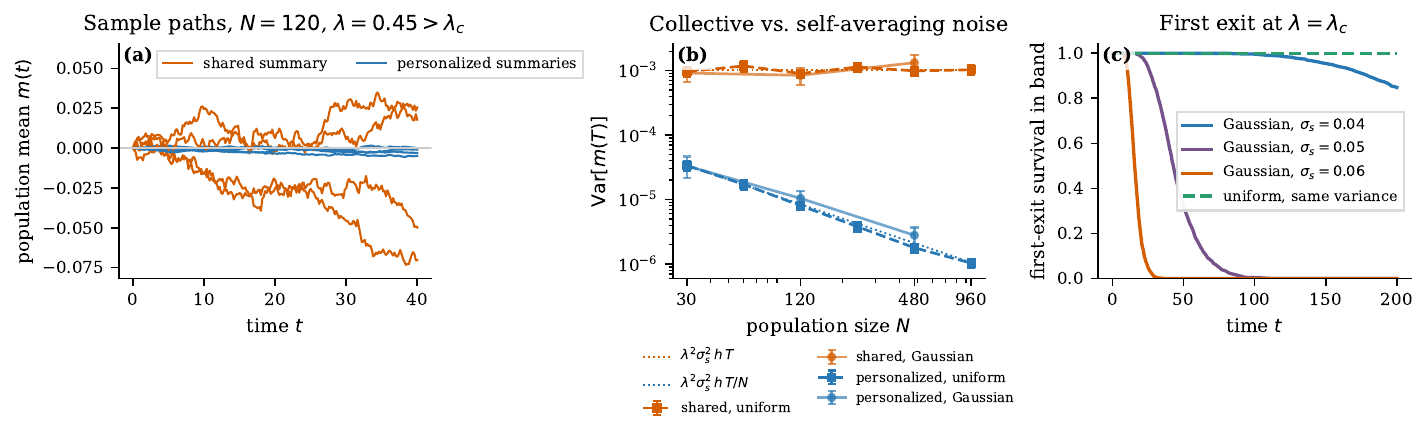}
  \caption{Deployment architecture and marginal gap-band first exit with
  $h=0.05$, $\sigma_s=0.05$, and projected step $0.01$ unless noted.
  (a) Population-mean paths for one shared error and independently
  personalized errors ($N=120$, $\lambda=0.45$).
  (b) Variance of $m(40)$ for bounded uniform attraction-cell errors (six
  sizes, 256 realizations per size) and Gaussian particle stress tests (three
  sizes, 64 realizations per size); bars are bootstrap 95\% intervals and
  dotted curves are the predicted scalings.
  (c) Survival in the gap band $[0.9,1.2]$ at $\lambda_c$ for 2000
  realizations per condition. Curves record first exit, not eventual
  destination.}
  \label{fig:regimes}
\end{figure*}

At marginality (Fig.~\ref{fig:regimes}, right), $85\%$ of the
$\sigma_s=0.04$ Gaussian runs remain in the gap band at $T=200$, whereas exit
is complete by approximately $115$ and $40$ for $\sigma_s=0.05$ and $0.06$.
Every first exit is outward, as expected when a tolerance violation removes
attraction while cross-bloc repulsion persists. The variance-matched bounded
control produces no exits in 2000 runs, isolating threshold-exceedance
probability rather than variance. Coarse and fine survival curves differ by at
most $1.5\times10^{-3}$, although rare pathwise hitting times remain sensitive;
no post-exit destination is inferred.

\subsection{Threshold crossings expose the limits of moment sufficiency}

The preceding collective reductions hold because the represented support stays
on a common linear branch. Figure~\ref{fig:force} closes the argument by
crossing that boundary in the minimal static counterexample:
the discussion and its mean-collapse representation generate identical linear
forces (distortion $1.1\times10^{-16}$), but under thresholded response the
positive bloc is pushed outward by the discussion ($+0.20$) and pulled inward
by the representation ($-0.50$). The exact one-sided supremum distortion is
$0.7700$ (the 2001-point grid maximum is $0.7693$) despite exact mean
preservation.

\begin{figure*}[t]
  \centering
  \includegraphics[width=\textwidth]{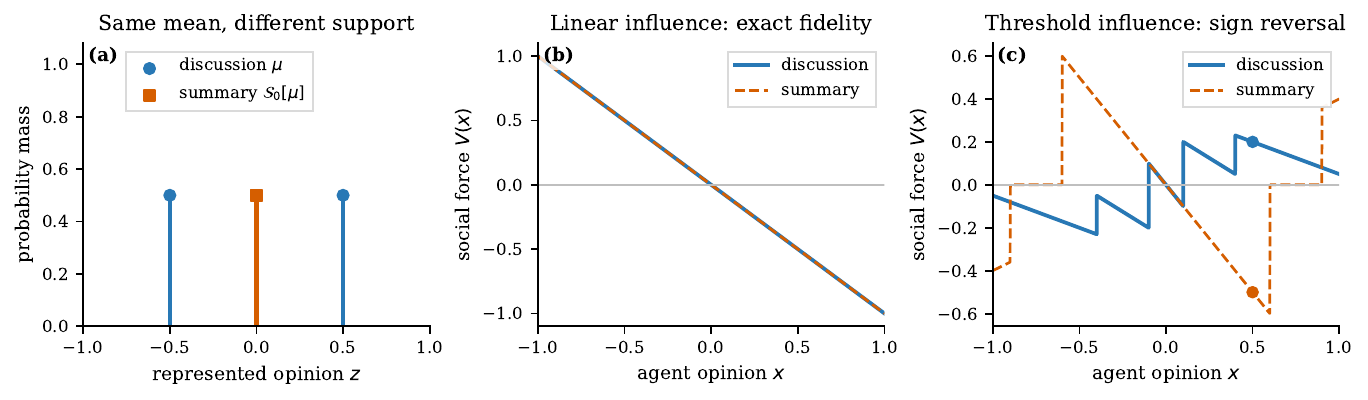}
  \caption{Force fidelity for
  $\mu=(\delta_{-0.5}+\delta_{0.5})/2$ and $\SF_0[\mu]=\delta_0$.
  (a) Equal means and different supports. (b) Identical force fields under
  linear response. (c) Opposite force directions at recipient $x=0.5$ under
  the threshold law with $(\epsilon_1,\epsilon_2,\eta)=(0.6,0.9,0.4)$.}
  \label{fig:force}
\end{figure*}

For the atomic threshold-crossing pair in Fig.~\ref{fig:threshold},
$\Wone$ falls to $2\times10^{-5}$ while the force error at $x=0$ remains
$0.360004$, equal to the jump limit $\eta\epsilon_2=0.36$ to four decimal
places. A translated smooth density with the same $\Wone$ has force error
$6.6\times10^{-5}$. This is the two-sided estimate of
Section~\ref{sec:transport}: transport distance controls the smooth part,
whereas cumulative mass crossing a response threshold controls the jump
part. Small transport error alone therefore does not define a dynamically
faithful coarse-graining.

\begin{figure*}[t]
  \centering
  \includegraphics[width=0.76\textwidth]{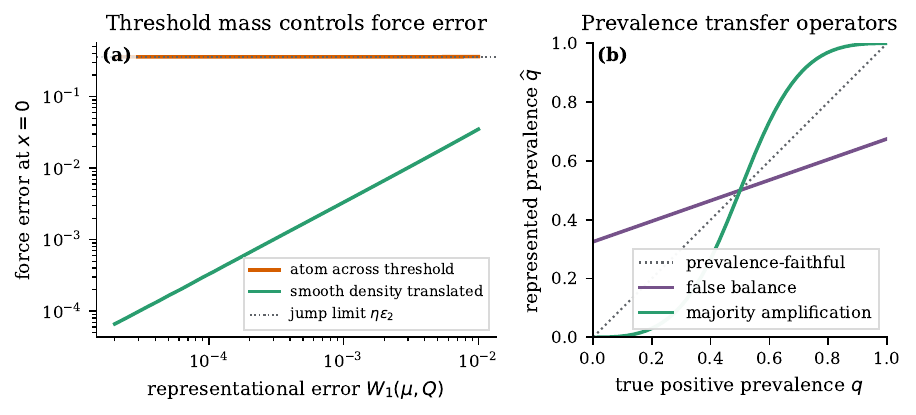}
  \caption{Threshold and prevalence transformations.
  (a) Force error at recipient $x=0$ when atomic or continuously distributed
  mass is displaced by $W_1$ across the outer response threshold; only the
  atomic sequence retains the jump-size error as $W_1\to0$.
  (b) Stylized false-balance ($c_q=0.35$) and majority-amplification
  ($a=2.5$) prevalence maps.}
  \label{fig:threshold}
\end{figure*}

\section{Empirical identification of summary operators}
\label{sec:calibration}

The abstract theory becomes a claim about any deployed summarization system
only after $\SF$ and the deployment architecture are measured. A three-stage identification design
separates the ingredients: generate discussions from known latent measures
(spanning prevalence, variance, skew, and mass near $\epsilon_1,\epsilon_2$,
with text realization randomized against confounds); summarize repeatedly
across models, prompts, and decoding seeds; and have independent raters map
outputs to perceived measures $Q$, yielding the kernel $K_\theta(\mu,\dd Q)$
and the operator-level quantities that the theory identifies as dynamically
relevant: prevalence transfer $q\mapsto\widehat q$, variance retention
$c^2$, threshold-layer transport, cumulative discrepancies at threshold
offsets (Theorem~\ref{thm:ks-lower}), and the regeneration variance
$\sigma_s^2$ together with whether outputs are broadcast or per-user
(Proposition~\ref{prop:deployment}).

The theory supplies kinetic signatures that do not require full
measure-valued calibration: the reversal at $\lambda_c$, the drift rate
$\lambda(\widehat q-q)\Delta$, which in a controlled fixed-cell experiment
with known $\lambda$ and $\Delta$ and other common drifts excluded estimates
the deployed operator's prevalence error, and the
$N$-(in)dependence of collective variance, which distinguishes the two
deployment architectures within the model; in field data this signature
would be confounded by other common shocks and must be interpreted jointly
with them. Human response must be estimated
separately from summary representation---randomized exposure experiments
testing whether updates are better predicted by $V_F[Q]$ than
$V_F[\mu]$---so that language models are used as the objects being measured,
never as substitutes for human subjects. Recursive experiments may be
especially sensitive near the spectral threshold, although optimal
experimental design is not analyzed here; this is also where ethical
safeguards matter most.

\section{Discussion}
\label{sec:discussion}

\subsection{Contribution hierarchy and scope}

The paper establishes a response-dependent map of behavioral sufficiency. On
a common attraction cell, the represented distribution is equivalent to its
mean for translation and is entirely invisible to disagreement. When probes
cross AIR thresholds, that equivalence collapses: translated threshold shadows
make generic finite supports identifiable, and cumulative mass transfer across
the shadows controls the part of force distortion that transport distance can
miss. These statements describe two regimes of the same endogenous
coarse-graining problem.

The primary model-specific results are therefore the exact response-cell
decomposition \eqref{eq:spectral-decomposition} and the bounded-probe
identifiability and inverse-stability theorems. The first derives a
summary-shape-independent disagreement operator while retaining represented-
mean error in translation; the second shows how the same response law recovers
support information once thresholds are visible and gives the explicit
support-reduction floor \eqref{eq:compression-floor}. The oriented cumulative
sandwich links the two regimes by isolating threshold-crossing mass.

The homogeneous spectral shift, the scalar two-bloc threshold, and the
shared-versus-independent self-averaging powers are not claimed as independent
mathematical innovations. They are consequences that translate the two main
characterizations into collective observables. At marginal stability,
prevalence bias produces secular translation only through its component along
a neutral disagreement mode; the two-bloc drift law is the complete-repulsion
specialization.

The deployment formulas have a deliberately narrower status. They are exact for common
translation errors in the point-bloc cell and for personalized errors applied
to a mean-preserving base operator while all represented positions remain
attractive. They do not establish a new common-noise principle or a
large-population limit. The six-size ensembles verify the resulting population-size
powers, but those powers are architectural self-averaging laws rather than
critical exponents. No thermodynamic phase transition or universality class
is claimed. Their use is operational: identical per-recipient error
marginals can produce different collective variance when their correlations
across recipients differ.

\subsection{Relation to existing work}

The work sits first in statistical-physics models of social dynamics
\citep{castellano2009,starnini2026}. Continuous bounded-confidence and
attraction--repulsion models supply the interacting population, while media
models supply fixed or distributed fields \citep{lorenz2007,sabinmiller2020,
vazmartins2010,pineda2015,pansanella2023}. Endogenous scalar AI signals
appear in oracle and aggregation models
\citep{rodrigo2025,acemoglu2026}. Reflexive institutions, platform-mediated
learning, predictive loops, and hierarchical coarse-graining establish that
endogenous feedback itself is not new
\citep{segoviamartin2021,candogan2022,wu2026,widler2026}. LLM influence has been incorporated into
bounded-confidence simulation \citep{li2026}; AI-mediated communication can
transform messages before social transmission \citep{tsirtsis2026}; and
recursive review summaries can bias social learning \citep{kim2026}.
Distribution-valued media inputs, global medians, fixed broadcasts,
recommenders, and LLM-agent populations cover adjacent mechanisms
\citep{kolarijani2021,berenbrink2026,mckeown2006,pansanella2023,
sirbu2013,sirbu2019,chuang2024}.
Accordingly, neither an endogenous AI signal nor anchored-network stability is
the novelty claimed here. The contribution is the response-dependent
sufficiency characterization: exact summary-induced mode separation inside a
common attraction cell and stable bounded-probe identification across AIR
thresholds. The neutral-mode bias criterion is a model-specific corollary of
the first characterization and standard linear spectral theory.

Signed-Laplacian spectral analysis, antagonistic consensus, signed
Friedkin--Johnsen stabilization,
and generalized solutions of discontinuous opinion systems are established
\citep{bronski2014,altafini2013,shi2019,razaq2025,ceragioli2012,altafini2018}.
In particular, the
centered homogeneous operator here coincides with the equal-stubbornness
signed-Friedkin--Johnsen transient operator. The new step claimed here is the
reduction of an endogenous distribution-valued summary channel to that
summary-shape-independent disagreement operator and the separation of
represented-mean error into translation. The associated neutral-mode
projection criterion follows transparently once that decomposition is known.
Similarly, shared versus independent noise is known in interacting systems
\citep{carmona2018,coculescu2024,vaidya2021}; the finite-cadence calculation
here is an application, not a claim of priority for common noise.

The fidelity analysis draws on a second literature. Task-relative information
value begins with comparison of statistical experiments \citep{blackwell1953}
and continues through task-oriented compression and decision-relative
evaluation \citep{guo2026,lee2026}. Integral probability metrics, transport
duality, and sparse-measure recovery are likewise established tools
\citep{muller1997,villani2009,vetterli2002,candes2014}. Our model-specific
inverse contribution is the threshold-shadow geometry created when the test class is
the translated social-response law and probes are restricted to the opinion
domain.

\subsection{Limitations and the path forward}

Opinions are one-dimensional and globally coupled, and summaries represent
stances rather than arguments. The response-cell theorem is local to an open
response cell, where the centered coordinates have a unique linear evolution.
Existence and uniqueness of the full trajectory additionally require
appropriate regularity of $x\mapsto\widehat m(x)$, which is not imposed here;
no global continuation is claimed across switching surfaces. The displayed
$(\lambda,c)$ surface is finite-horizon numerical
evidence. Boundary conditions are known to affect noisy bounded-confidence
phases \citep{goddard2022}; our projected convention is therefore part of the
model, not an innocuous implementation detail.

The inverse theorem covers separated finite supports. Non-generic threshold
lattices and non-atomic discussions remain unresolved. Personalized
prevalence distortion falls outside the exact variance coefficient because
independent errors randomize the shape-dependent systematic drift. The
Gaussian and opinion-diffusion runs are stress tests that can leave the strict
cells assumed by the proofs.

Most importantly, no language model or human response law is calibrated.
The prevalence maps and error levels are stylized. A preregistered study must
measure both $K_\theta(\mu,\dd Q)$ and recipient response before the model can
make claims about a deployed system.

A substantial remaining mathematical problem is a conditional
propagation-of-chaos and selection theorem for projected particles driven by
an endogenous random summary kernel under the discontinuous response law.
Common-noise limits and discontinuous bounded-confidence systems are known
separately; a useful advance would need to handle their combination and prove
that the numerical selection converges. That problem is not solved here. The
present contribution is narrower and explicit: it characterizes behavioral
sufficiency in two complementary regimes, through exact endogenous-summary
mode separation on strict response cells and bounded-probe identifiability and
stability for generic finite atomic discussions across the discontinuities.
The prevalence-bias projection criterion is a corollary of the fixed-cell
reduction; the scalar stabilization shift itself is not claimed as new.

\section*{Data and code availability}

The repository contains the source code, tests, figure-generation script,
aggregate results, raw population-size ensembles, mode trajectories, first-exit
records, refinement arrays, and a machine-readable provenance manifest at
\url{https://github.com/REsteche/social-modeling-2}. No proprietary model API
or external dataset is required to reproduce the reported computations. The
manifest records software versions, random seeds, and cryptographic hashes of
the generated artifacts.

\section*{Credit authorship contribution statement}

Ruben E. Ara\'ujo: Conceptualization, Methodology, Formal analysis, Software,
Validation, Investigation, Visualization, Writing--original draft,
Writing--review \& editing.

\section*{Funding}

This research did not receive any specific grant from funding agencies in
the public, commercial, or not-for-profit sectors.

\section*{Declaration of competing interest}

The author declares no known competing financial interests or personal
relationships that could have appeared to influence the work reported here.

\section*{Declaration of generative AI and AI-assisted technologies}

During preparation of this work, the author used Anthropic Claude and OpenAI
ChatGPT for literature discovery, code development, and drafting assistance.
The author subsequently reviewed and edited the manuscript, checked the
proofs, references, and numerical results, and takes full responsibility for
the published content.

\bibliographystyle{elsarticle-num-names}
\bibliography{references}

\end{document}